\documentclass[sigplan, A4]{acmart}

\renewcommand\footnotetextcopyrightpermission[1]{}
\usepackage[utf8]{inputenc}
\usepackage{xspace}
\usepackage{listings}
\usepackage{hyperref}
\usepackage{color}
\usepackage{amsmath}
\usepackage{amsfonts}
\usepackage{graphicx}
\usepackage{url}
\usepackage{colortbl}
\usepackage{multirow}
\usepackage{pifont}
\usepackage{subfigure}
\usepackage{macros}
\newtheorem{theorem}{Theorem}
\newtheorem{lemma}{Lemma}
\newtheorem{definition}{Definition}
\definecolor{black}{RGB}{0,0,0}
\definecolor{gray}{RGB}{102,102,102}        
\definecolor{function}{RGB}{0,102,153}      
\definecolor{lightgreen}{RGB}{102,153,0}    
\definecolor{lightlightgreen}{RGB}{152,193,50}    
\definecolor{bluegreen}{RGB}{51,153,126}    
\definecolor{magenta}{RGB}{217,74,122}  
\definecolor{orange}{RGB}{226,102,26}       
\definecolor{purple}{RGB}{125,71,147}       
\definecolor{green}{RGB}{113,138,98}        
\definecolor{tomato}{RGB}{255,99,71}  
\definecolor{lightred}{RGB}{255,160,131}  

\lstdefinelanguage{parameterized}{
  firstnumber=1,
  xleftmargin=1.5em,
  numberstyle=\tiny\color{black},
  tabsize=2,
  numbers=left,
  morekeywords = [3]{require,while,if,then,else,do,done,wait,until,end,for,return,returns,upon,from,to,is,in},
  morekeywords = [4]{pragma,function,contract,new,true,false,null,and,or},
  morekeywords = [5]{bench,unlockAccount,createBatch,start_new_consensus,peek,
  propose,execute_transaction,HttpProvider,parse,on,catch,readFileSync,sendSignedTransaction,
  consensus_propose,mvc_propose,break,decide,poll,broadcast,deliver,add,greet,Hello,
  execute,sendTransaction,updateBlockState,executeTx,persist,deserialize,isValid,add},
  morekeywords = [6]{var,public,bytes32,bool,byte,+,=,:=,.,;,,,-,!,=,~,>,<,==,solidity},
  morekeywords = [7]{reliable_broadcast,binary_consensus,consensus},
  keywordstyle = [3]\color{bluegreen},
  keywordstyle = [4]\color{lightgreen},
  keywordstyle = [5]\color{magenta},
  keywordstyle = [6]\color{orange},
  keywordstyle = [7]\color{purple},
  sensitive = true,
  morecomment = [l][\color{gray}]{//},
  morecomment = [s][\color{gray}]{/*}{*/},
  morecomment = [s][\color{gray}]{/**}{*/},
  morestring = [b][\color{purple}]",
  morestring = [b][\color{purple}]',
  literate=
  {á}{{\'a}}1 {é}{{\'e}}1 {í}{{\'i}}1 {ó}{{\'o}}1 {ú}{{\'u}}1
  {Á}{{\'A}}1 {É}{{\'E}}1 {Í}{{\'I}}1 {Ó}{{\'O}}1 {Ú}{{\'U}}1
  {à}{{\`a}}1 {è}{{\`e}}1 {ì}{{\`i}}1 {ò}{{\`o}}1 {ù}{{\`u}}1
  {À}{{\`A}}1 {È}{{\'E}}1 {Ì}{{\`I}}1 {Ò}{{\`O}}1 {Ù}{{\`U}}1
  {ä}{{\"a}}1 {ë}{{\"e}}1 {ï}{{\"i}}1 {ö}{{\"o}}1 {ü}{{\"u}}1
  {Ä}{{\"A}}1 {Ë}{{\"E}}1 {Ï}{{\"I}}1 {Ö}{{\"O}}1 {Ü}{{\"U}}1
  {â}{{\^a}}1 {ê}{{\^e}}1 {î}{{\^i}}1 {ô}{{\^o}}1 {û}{{\^u}}1
  {Â}{{\^A}}1 {Ê}{{\^E}}1 {Î}{{\^I}}1 {Ô}{{\^O}}1 {Û}{{\^U}}1
  {Ã}{{\~A}}1 {ã}{{\~a}}1 {Õ}{{\~O}}1 {õ}{{\~o}}1
  {œ}{{\oe}}1 {Œ}{{\OE}}1 {æ}{{\ae}}1 {Æ}{{\AE}}1 {ß}{{\ss}}1
  {?}{{\H{u}}}1 {?}{{\H{U}}}1 {?}{{\H{o}}}1 {?}{{\H{O}}}1
  {ç}{{\c c}}1 {Ç}{{\c C}}1 {ø}{{\o}}1 {å}{{\r a}}1 {Å}{{\r A}}1
  {€}{{\euro}}1 {£}{{\pounds}}1 {«}{{\guillemotleft}}1
  {»}{{\guillemotright}}1 {ñ}{{\~n}}1 {Ñ}{{\~N}}1 {¿}{{?`}}1
}

\newcommand{\cref}[1]{{\S\ref{#1}}}
\newcommand{\solution}{SRBB\xspace}
\newcommand{\solutionlong}{Smart Red Belly Blockchain\xspace}

\newcommand{\middleware}{middleware\xspace}

\newcommand{\remove}[1]{}
\newcommand\Smallish{\fontsize{7}{8.8}\selectfont}
\newcommand*\LSTfont{\Smallish\ttfamily\SetTracking{encoding=*}{-60}\lsstyle}
\newcommand\TextSize{\fontsize{8.5}{9.5}\selectfont}
\newcommand*\ttt{\TextSize\ttfamily\SetTracking{encoding=*}{-60}\lsstyle}

\author{Deepal Tennakoon and Vincent Gramoli}
\date{}

\begin{document}
\title{\solutionlong: Enhanced Transaction Management for Decentralized Applications}

\begin{abstract}

Decentralized Applications (DApps) have seen widespread use in the recent past driving the world towards a new decentralized version of the web known as Web3.0. DApp-supported blockchains like Ethereum have largely been responsible for this drive supporting the largest eco-system of DApps. Although the low performance provided by Ethereum has been a major impediment to realizing a decentralized web, several high-performance blockchains have been introduced recently to bridge this gap.
Most of these blockchains rely on consensus optimizations.
Only a few enhance other parts of the blockchain protocol that involves transaction management: the validation of transactions, broadcast of transactions, encapsulation and dissemination of blocks with transactions, re-validation and execution of transactions in blocks, storage of blocks, and confirmation of transaction commits to senders upon request.

In this paper, we enhance transaction management by introducing a novel transaction validation reduction and a per sub-block processing to optimize the block storage. We empirically show the performance improvements gained by our enhanced transaction management in the \solutionlong (\solution) VM we develop. Finally, we integrate our \solution VM to an already optimized consensus from a known blockchain to develop the \solutionlong. Our results show that \solution achieves a peak throughput of 4000\,TPS and an average throughput of 2000\,TPS on 200 nodes spread across 5 continents. \solution outperforms 6 other blockchains when running the exchange DApp featuring a real workload trace taken from \textsc{Nasdaq}.
\end{abstract}
\maketitle

\section{Introduction}
As the number of misuses of Internet data grows due to the centralization of the Web~\cite{MSH16}, so does the need for decentralized \middleware rewarding individuals for sharing data.
This centralization has severe drawbacks: it exposes data to leaks and hacks~\cite{Fac18} and
it facilitates user manipulation~\cite{Pra18}.
Although Ethereum~\cite{wood2014ethereum} promised to decentralize the Web by offering decentralized applications (DApps) written as smart contracts, it remains inherently slow. 
The crux of the problem stems from its transaction management~\cite{ponnapalli2021rainblock}.

The classic way of managing transactions (i.e., transaction management) in Ethereum~\cite{wood2014ethereum} and Bitcoin~\cite{Nak08} consists of miners validating each transaction upon reception, disseminating a transaction to other validators,  encapsulating transactions in a block, disseminating this block to other validators, re-validating and executing the block transactions~\cite{wood2014ethereum, libra21}, storing blocks and finally the sender confirming that the transaction is committed.
This process has major shortcomings such as every validator having to validate each transaction twice.
This is the reason, one needs an \emph{enhanced transaction management} to speedup blockchains supporting DApps.

While various efforts focus on consensus optimizations~\cite{YMR19,CNG21,CMMP21,RML21},
fewer efforts focus on enhancing the transaction management~\cite{yakovenko2018solana, ponnapalli2021rainblock, https://doi.org/10.48550/arxiv.2203.06871} in blockchains. Some of these efforts trade-off security for improved performance~\cite{yakovenko2018solana}~\cite{ponnapalli2021rainblock} while others work only in test environments~\cite{273865}.


In this paper, we enhance transaction management by introducing a novel transaction validation reduction and a per sub-block storage optimization for storage of blocks.

First, we divide the number of necessary validations by two by not requiring all validators to eagerly validate every transaction when it is submitted but by requiring all miners to only lazily validate this transaction upon execution.
By validating twice less we reduce the validation time of Ethereum by up to 48\% when the number of validators is large. Second, we process and store sub-blocks in a block yielding a throughput improvement of 39.7\%.

To illustrate the applicability of our optimizations, we implemented our enhanced transaction management with the consensus of the Red Belly Blockchain~\cite{CNG21}, a fast blockchain that does not support DApps (i.e. smart contracts), to obtain the \solutionlong (\solution)
that supports the largest ecosystem of DApps.
We deploy \solution in a geo-distributed setting of 200 machines located across 5 continents.
We compare its performance on the \textsc{Nasdaq} application of the Diablo benchmark~\cite{BGG21}, featuring the
real trace of Apple, Amazon, Facebook, Microsoft, and Google stock trades.
Not only does \solution commit all transactions but it outperforms Algorand~\cite{GHM17}, Avalanche~\cite{TR18}, Ethereum~\cite{wood2014ethereum}, Libra-Diem~\cite{Ams20}, Quorum~\cite{qwp} and Solana~\cite{Yak21}.
Finally, \solution achieves a peak throughput of 4000\,TPS and an
average throughput of \textasciitilde{}2000\,TPS. In summary, our technical contributions include the following:
\begin{itemize}
\item To reduce the CPU usage of the EVM, we reduced transaction validation time by 48\% without weakening security (Section~\ref{ssec:tx-validation}).

\item To reduce the I/O delays of the EVM, we broke large blocks into sub-blocks to be processed one at a time (Section~\ref{superblockopt-section}), hence improving throughput by 39.7\%.

\item To illustrate the benefit of these optimizations in a realistic setting, we integrated this enhanced transaction management into a blockchain combining the optimized  EVM (i.e., \solution VM) with the consensus of the Red Belly Blockchain~\cite{CNG21} to obtain \solutionlong. Not only does \solution show a $3.2\times$ speedup in throughput over the naive blockchain with the EVM and the Red Belly Blockchain consensus simply plugged together (Section~\ref{sec:cumulativeperf}), but it also outperforms 6 state-of-the-art blockchains when running a realistic exchange DApp (Section~\ref{section:comparison}).
\end{itemize}

In the remainder of the paper, we present our background and assumptions (Section~\ref{sec:background}). We then present \solution transaction management along with our optimizations (Section~\ref{sec:solution}). Next, we evaluate our transaction management enhancements and the resulting \solution against other recent blockchains (Section~\ref{sec:evaluationSRBB}).
Finally, we present the related work (Section~\ref{sec:rw}) and conclude (Section~\ref{sec:conclusion}). The consensus protocol and the proof of correctness are presented in Appendix~\ref{line:proof}.

\section{Background and Assumptions}
\label{sec:background}

In this section, we present our background. Next, we define the blockchain problem and present our assumptions.

\subsection{Ethereum background}
Ethereum~\cite{wood2014ethereum} features the Ethereum Virtual Machine (EVM) that was proposed in part to cope with the limited expressiveness of Bitcoin~\cite{Nak08} and to execute DApps written in a Turing complete programming language as smart contracts. Go Ethereum, or {\ttt geth} for short, is the most deployed Ethereum implementation~\cite{gethstats}.

\paragraph{Ethereum nodes, accounts and transactions}
Ethereum \emph{nodes} are of two main types, namely client nodes and validator nodes. Client nodes send read and write requests to the blockchain while a \emph{validator} node (i.e., miner) services these requests. Note that the term ``client'' is used to define implementations of Ethereum (e.g., {\ttt geth} client) by the Ethereum community but we identify a ``client'' solely as a sender of requests to the blockchain. More precisely, validator nodes perform two tasks. Firstly, the validator solves consensus to agree upon the client write requests and executes those requests. Secondly, the validator services client read requests. Clients can have \emph{accounts} in Ethereum. An Ethereum account contains a nonce (the number of write requests sent from this account), an address (the identifier of the account derived from the associated public key), and a balance (the amount of funds in the account). \emph{Transactions} are write requests sent by client nodes to the blockchain. They are of three main types in Ethereum: native payments that transfer funds between ethereum accounts, smart contract deployments that upload smart contracts to the blockchain, and smart contract invocations that invoke the functions of smart contracts. The transaction execution costs some fee to the client. This fee was originally calculated from gas -- the amount of computational work required to execute a transaction and the gas price -- a monetary value expressed in a cryptocurrency for one unit of gas~\cite{wood2014ethereum}. Thus, the fee calculation was measured by gas $\times$ gas price~\cite{wood2014ethereum}, even though a gas tip has recently been integrated into the fee calculation.


\paragraph{The redundant validations of Ethereum}
\label{sec:bgrd}

In order to check that a request (or transaction)
is valid, all of the {\ttt geth} validators must validate twice each transaction prior to execution:
\begin{itemize}
\item {\bf Eager validation:} This validation occurs upon reception of a new client transaction and verifies the transaction size does not exceed 32 kilobytes, the value transferred is non-negative, the gas of a transaction does not exceed the block gas limit, the transaction is properly signed, the transaction nonce value is in order, the sender account has sufficient coins and the gas amount is sufficient to execute the transaction. Because of this transaction validation, the risk of denial-of-service (DoS) attacks is reduced as an invalid transaction is dropped early. If the transaction is valid, {\ttt geth} propagates it to other nodes.
\item {\bf Lazy validation:} This validation occurs before the transactions in a block are executed and simply checks the nonce and whether there is enough gas for execution.
This lazy validation is necessary to guarantee that transactions in a newly received block are indeed valid. The lazy validation is thus less time-consuming in {\ttt geth} than the eager validation. This is why we focus here on reducing the number of eager validations.
\end{itemize}
Note that this double validation is overly conservative because each to-be-executed transaction of {\ttt geth} is validated twice by each server. This is unnecessary as an invalid transaction coming from a malicious node will either be dropped by the lazy validation before its execution or fail to execute. It is also interesting to note that in a system with few malicious replicas, there is no need for all nodes to validate all transactions twice. We explain in Section~\ref{ssec:tx-validation} how, without reducing security, we reduce the number $k$ of eager validations per validator node down to $k/n$ to scale to a large system of size $n$.

\paragraph{Transaction blocks and receipts}
A transaction block consists of a list of transactions. In the EVM, transaction blocks are executed after a block is followed by a branch of 5 additional blocks (i.e., this is known as block confirmation). After execution, the blocks are inserted into a chain where each block is linked to the prior block forming a chain of blocks, hence the name ``blockchain''. Subsequently, the blocks are written to the LevelDB (a key-value data store). Once the data in the LevelDB exceeds a threshold of memory, the blocks stored in memory are flushed to the disk using disk I/O. The transaction receipts record the result of executing a transaction. It contains the status (a boolean value indicating whether or not the transaction execution was successful), a block hash (the hash of the block containing the transaction), a block number (the index of the block containing the transaction), and transaction-specific fields. These include the transaction hash, sender address and receiver address (a null value if it is a smart contract deployment). Next, the receipt contains additional fields such as the contract address (the address of the smart contract created if it is a contract deployment), cumulative gas used (the total gas consumed by the blockchain when the transaction was executed), the gas used (the amount of gas used by the transaction) and an array of logs (these logs contain events generated by the transaction). 

\paragraph{Transaction management}
The transaction management is the process transactions go through when a transaction is submitted to the blockchain by a client until the client confirms the transaction has been committed. Therefore, the typical Ethereum transaction management includes the reception of transactions at the blockchain, the validation of transactions, the broadcast of transactions to the network, the encapsulation of transactions into a block, the dissemination of the blocks among validators, the re-validation of transactions in a block, the execution of transactions in a block, the storage of blocks and finally the retrieval of transaction data by the client to verify that transactions are committed.
 
\paragraph{Ethereum tries}
Ethereum transactions executed on validator nodes update the Ethereum state -- a mapping of account addresses and account state (i.e., account balances, account nonces and smart contract variable states). The Ethereum state, transactions and transaction receipts are stored in a Merkel Patricia Trie (MPT) data structure~\cite{wood2014ethereum} first held in memory and then flushed to disk once a threshold of memory exceeds. The roots of these three tries (i.e., state, transaction and receipt tries) reside in the header fields of transaction blocks as the state root, transaction root and transaction receipt root. Therefore, each block points to its respective snapshot of the three tries allowing anyone to traverse through the blocks to view the state, transactions and receipts history.

\paragraph{Transaction confirmation}
A transaction is confirmed from the transaction issuer's (i.e., sending client) standpoint when the issuer verifies the transaction as committed. We term this process as transaction confirmation.
Usually, any client including the transaction issuer knows a particular transaction is committed by fetching the transaction receipt (i.e., receipt polling) from the blockchain or by fetching a block containing the transaction (i.e., block polling). Therefore, when measuring latency and throughput of blockchains, evaluation tools~\cite{BGG21}~\cite{DWC17} consider the end-to-end measurement between the time when the transaction is sent from the issuer to the time the transaction receipt or the block containing the transaction is received back at the issuer.

\subsection{Assumptions and blockchain problem}
Our network consists of validator nodes that are well-connected. We assume \emph{partially synchronous} communication~\cite{DLS88}. 
We assume, out of $n$ \solution validator nodes, at most $f$ are \emph{byzantine} and can act arbitrarily where $f < n/3$. As outlined in Section~\ref{ssec:open}, \solution requires nodes to deposit stake to become a validator. The use of PoS (Proof-of-Stake) provides a form of Sybil resistance as it makes it expensive for a user to join the network with different identities as multiple validators. Also mentioned in Section~\ref{ssec:open}, \solution features a reconfiguration protocol to cope with bribery.

\label{ssec:bc}
We refer to the blockchain problem as the problem of ensuring both the safety and liveness properties that were defined in the literature by Garay et al.~\cite{GKL15} and restated more recently by Chan et al.~\cite{CS20}, and a classic validity property~\cite{CNG21}.
\begin{definition}[The Blockchain Problem]\label{def:blockchain}
The \emph{blockchain problem} is to ensure that a distributed set of validator nodes
maintain a sequence of transaction blocks such that the following properties hold:
\begin{itemize}
\item  \emph{Liveness:} if a correct validator node receives a transaction, then this transaction will eventually be reliably stored in the block sequence of all correct validator nodes.
\item \emph{Safety:} the two chains of blocks maintained locally by two correct validator nodes are either identical or one is a prefix of the other. 

\item \emph{Validity:} each block appended to the blockchain of each correct validator node is a set of valid transactions (non-conflicting well-formed transactions that are correctly signed by its issuer).
\end{itemize}
\end{definition}
The safety property does not require correct validator nodes to share the same copy, simply because one replica may already have received the latest block before another receives it.
Note that, as in classic definitions~\cite{GKL15,CS20}, the liveness property does not guarantee that a client transaction is included in the blockchain: if a client sends its transaction request exclusively to byzantine nodes then byzantine nodes may decide to ignore it.

\section{\solution with Enhanced Transaction Management}
\label{sec:solution}
In this section, we present \solution, our blockchain compatible with the largest ecosystem of DApps. 
It is optimally resilient against byzantine failures. Firstly, we present the architecture of \solution (Figure~\ref{fig:architecture}) and then describe its functionality through the transaction life cycle. Subsequently, we present our transaction management enhancements for the \solution VM under Sections~\ref{ssec:tx-validation}, \ref{superblockopt-section}. Finally, we discuss, the membership of \solution.

\begin{figure}[t]
\begin{center}
\includegraphics[scale=0.5]{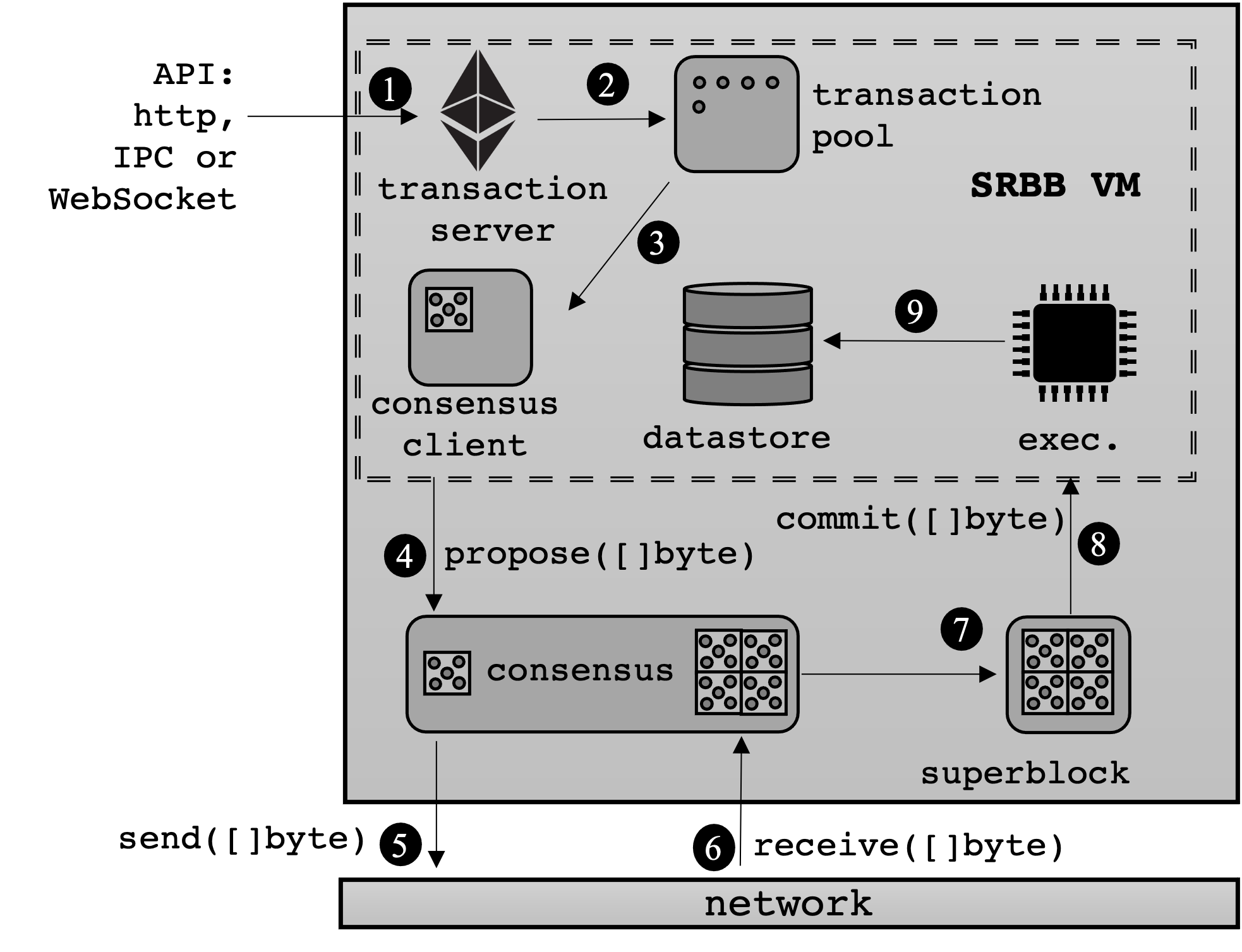}
\caption{
The architecture of a \solution node. The \solution VM is built from Geth enhancing its transaction management. The consensus is the Red Belly Blockchain's consensus~\cite{CNG21}. \ding{202} A client sends a transaction to some replica(s), 
\ding{203}
at each replica, the transaction server receives the transactions and sends them to the transaction pool that validates transactions and
\ding{204} sends a block to the consensus client. \ding{205} The consensus client $\lit{propose}$s it to the consensus protocol. \ding{207} When the consensus outputs some acceptable blocks, \ding{208} all of these blocks are combined into a superblock and sent for execution. Each block in the superblock is validated and executed\ding{209}. Subsequently, the block is inserted into a chain and stored in the data store\ding{210}.
\label{fig:architecture}
}
\end{center}
\end{figure}



\subsection{The transaction lifecyle}\label{sec:lifecycle}

The lifecycle of a transaction goes through these subsequent stages: 
\begin{itemize}
\setlength{\itemindent}{.56in}
\item[{\bf 1. Reception.}] 
The client creates a properly signed transaction and sends it to at least one \solution node \ding{202}. Once a request containing the signed transaction is received \ding{203} by the transaction server of the \solution node, the transaction in the request is submitted to the transaction pool. In the transaction pool, the eager validation (\ref{sec:bgrd}) starts. 
If the validation fails, the transaction is discarded.
If the validation succeeds, the transactions are kept in the transaction pool. Unlike in Ethereum where the transaction would be propagated 
to all miners increasing the number of eager validations, \solution simply proposes it to the consensus as follows:
If the number of successfully validated transactions in the transaction pool reaches a threshold, then the transaction pool creates a new proposed block with a number (defined by the threshold) of transactions from the pool.
It serializes and sends the proposed block to the consensus client \ding{204}. 

\item[{\bf 2. Consensus.}] 
Once the consensus client receives a proposed block, it sends the corresponding byte array to the consensus protocol by invoking the {\ttt propose([]byte)} method \ding{205}. The consensus sends \ding{206} and receives \ding{207} proposals to reach agreement. Subsequently, the consensus creates a superblock \ding{208} with all acceptable blocks/proposals (The Consensus Protocol, line~\ref{line:superblock}) and sends this superblock to the \solution VM by invoking the {\ttt commit([]byte)} method \ding{209}.

\setlength{\itemindent}{.43in}
\item[{\bf 3. Commit.}] 
When the superblock is received by the \solution VM \ding{209} the following execution process follows: a block is taken at a time from the superblock and deserialized using JSON unmarshalling. Then the \solution VM does the lazy validation (\ref{sec:bgrd}) of each transaction in the block. Note that as opposed to the eager validation, all \solution nodes execute the lazy validation before executing a transaction. Yet it does not prevent \solution from scaling to hundreds of nodes (\ref{sec:geodistributed}). After lazy validation, the transactions in the block are executed and the state, transaction, and transaction receipt tries are updated in memory adding the root of these tries to the block header. The block is then inserted/written to the chain in the data store with a pointer to the previous block \ding{210}. Next, the \solution VM follows the same procedure to process the subsequent block in the superblock until the entire superblock is committed. Note that, at some point in time, when a threshold of memory is reached the data store held in memory is flushed to the disk.
\end{itemize}

\subsection{Reducing the transaction validations}\label{ssec:tx-validation}

As opposed to each Ethereum validator node that validates 
eagerly and lazily each of the $k$ transactions of the system, each of the $n$ \solution nodes eagerly validates on average $k/n$ transactions.
Specifically, only one \solution node needs to eagerly validate each transaction: the first node receiving the transaction validates it but does not propagate it to other nodes and simply proposes it to the consensus. 
As a result, \solution limits the redundant validations reducing the CPU overhead, which improves performance. More precisely, if the number of \solution nodes is $n$, then each node does $1+1/n$ validations per transaction on average (one lazy validation + $1/n$ eager validation) compared to the two validations needed in {\ttt geth}. As $n$ tends to infinity, \solution nodes validate on average half of what {\ttt geth} nodes validate leading to a performance increase of 48\% (Section~\ref{sec:vreduction}).
In the worst case, where all clients send their transactions to $f+1 = n/3$ \solution nodes simultaneously, then each node will still eagerly validate only $k/3$ transactions.

Note that, as a result of our optimization, a byzantine validator node could propose transactions to the consensus without validating them eagerly, in this case, two things can happen: (i) The transaction is discarded at the lazy validation if invalid (ii) The validator node attempts to execute the invalid transaction, fails at execution and reverses the state to what it was. Either way, there is no impact on the safety of the blockchain. This is also not a DoS vulnerability of \solution, as even a byzantine node in Ethereum can propagate invalid transactions to all nodes, forcing unnecessary eager-validations.

\subsection{Per sub-block processing}
\label{superblockopt-section}
Since our consensus system is fast, it creates and delivers superblocks at high frequency through the commit channel to the \solution VM. As {\ttt geth} does not expect to receive blocks at such a high frequency, it raises an exception outlining that consecutive block timestamps are identical, which never happens in a normal execution of Ethereum.
This equality arose because  {\ttt geth} encodes the timestamp of each block as {\ttt uint64}, not leaving enough space for encoding time with sufficient precision. Typically {\ttt geth} reports an error when consecutive timestamps are identical, due to a strict check that compares the parent block timestamp to the current block timestamp:
{\ttt header.Time < parent.Time}.

We changed the original check to {\ttt header.Time <= parent}, which allowed for fast-paced executions of consecutive blocks. At each index of the blockchain, the \solution node typically executes many more transactions as part of function {\ttt execute\_transaction} (lines~\ref{line:exec-tx-start}--\ref{line:exec-tx-end}) than Ethereum. 
This is due to a single consensus instance outputting a superblock through the {\ttt commitChan} channel which contains potentially as many blocks as \solution nodes (line~\ref{line:commit-chan}).
\begin{lstlisting}[language=parameterized,basicstyle=\LSTfont,escapechar = ?,escapeinside={(*}{*)},frame = single]
execute_transaction: (*\label{line:exec-tx-start}*)
  // for each superblock received 
  for superblock in node.commitChan do (*\label{line:commit-chan}*)
    vtxs := (*$\emptyset$*) // set of valid transactions
    for block in superblock do // each block (*\label{line:block-loop-start}*)
      txs := node.txm.deserialize(block) // get txs
      for tx in txs:
        if isValid(tx): vtxs.add(tx) // lazy validation (*\label{line:validate}*)
      // set the corresponding block state and order txs
      updateBlockState() (*\label{line:update-block-state}*)
      for tx in vtxs do // for each valid tx...
        executeTx(tx) // ...execute it
      done
      persist(vtxs) // persist valid txs to disk (*\label{line:persist}*)
      vtxs := (*$\emptyset$*)
    done  (*\label{line:block-loop-end}*)
  done (*\label{line:exec-tx-end}*)
\end{lstlisting}

Due to this excessive execution of transactions, the \solution nodes consume high CPU.
Typically, high CPU usage slows down the processing of transactions in the \solution node, which results in the increase of pending transactions in the transaction pool stored in-memory increasing the memory usage.
Due to high memory usage, the garbage collection gets activated to gain memory space. Also, when the memory usage is high the \solution node flushes Ethereum tries to disk to save memory causing significant I/O. This further slows down the \solution VM.

As a solution, we optimized the \solution VM to fully process one proposed block (i.e. sub-block) of the superblock at a time allowing it to alternate frequently between CPU-intensive (verifying signatures and transaction executions) and memory-intensive (state trie write, receipt/transaction tries writes) and IO-intensive tasks (flushing to LevelDB). This resulted in a throughput improvement of 39.7\% (Section~\ref{ssec:store}).

\subsection{Proof-of-stake and membership change}\label{ssec:open}

\solution is an open blockchain that does not require permission for any node to join the network. An \solution node can play two roles, either as a client that sends transactions and reads the state of the blockchain or as a validator node that participates in consensus, executes transactions and keeps a full state of the ledger to service clients. To be a validator, a \solution node must deposit some tokens to a reconfiguration smart contract as stake after which the smart contract outputs a random committee of \solution nodes based on stake to service client requests similar to a sortition~\cite{GHM17}. The \solution nodes listen to this smart contract committee output, which is a smart contract event with a set of IDs (static IP addresses) grouped into a committee. Then the
\solution validator nodes that are part of the committee receive a reward as an incentive for deciding blocks. After a specific time (i.e., an epoch), the validator \solution nodes in a committee invoke a function in the reconfiguration smart contract triggering a committee change similar to the way the initial committee was setup. The use of stake helps mitigate Sybil attacks while committee rotation helps against bribery attacks from slowly-adaptive adversaries. A detailed implementation of the membership, reconfiguration and block reward incentives is out of the scope of this paper and we leave it as future work.

\section{Evaluation}
\label{sec:evaluationSRBB}
In this section, we present our evaluation of \solution. In summary, our transaction management enhancements result in \solution reaching a peak throughput of \textasciitilde{}4000\,TPS and an average throughput of 2000\,TPS for 200 nodes located in 5 continents, and outperforming 6 state-of-the-art blockchains when executing the Diablo benchmark DApp with the real \textsc{Nasdaq} workload~\cite{BGG21}.
Then, we show the overall performance improvement when we combine our optimizations (Section~\ref{sec:cumulativeperf}). Next, we evaluate the scalability of our blockchain (Section~\ref{sec:geodistributed}). Finally, we compare \solution with state-of-the-art recent high performance blockchains (Section~\ref{section:comparison}).

Our evaluation focuses on the latency and throughput performance metrics. The latency of a transaction is the elapsed time between the transaction send time and the transaction confirmation time as seen by the sending client. The throughput is the number of transactions committed per second as observed by the client (recall that \solution ensures instant finality so the transaction is committed as soon as it is stored in reliable storage). To display throughput time series (Figures~\ref{fig:optimized}, \ref{fig:srate} and \ref{fig:gafam}) and avoid the sawtooth effect produced by blockchains appending transactions in blocks, we
averaged the throughput at each second as the average throughput observed within a sliding window of three seconds.
Note that, we evaluate all blockchains using the Diablo~\cite{BGG21} tool that evaluates blockchains against specified workloads by sending pre-signed transactions. For geo-distributed experiments, we use 10 AWS regions spanning 5 continents. Namely: Bahrain, Cape Town, Milan, Mumbai, N. Virginia, Ohio, Oregon, Stockholm, Sydney, and Tokyo. For all benchmarks, we use the AWS family of c5 EC2 instances.

\label{line:eval}

\subsection{Effects of \solution per sub-block processing}\label{ssec:store}
We compared the naive method which executes transactions in an entire superblock and commits them, with an optimized method of processing each sub-block in a superblock at a time (Section~\ref{superblockopt-section}). Note that, to form a superblock we require \solution to execute more than one node as it requires consensus from multiple nodes to deliver a superblock to the \solution VM. Since a BFT consensus requires at least four nodes, we evaluated the \solution superblock processing on a c5.2xlarge 4 node setup with 15000 transactions.

\begin{figure}[t]
	\centering
	\includegraphics[scale=0.29,clip=true,viewport=20 0 900 220]{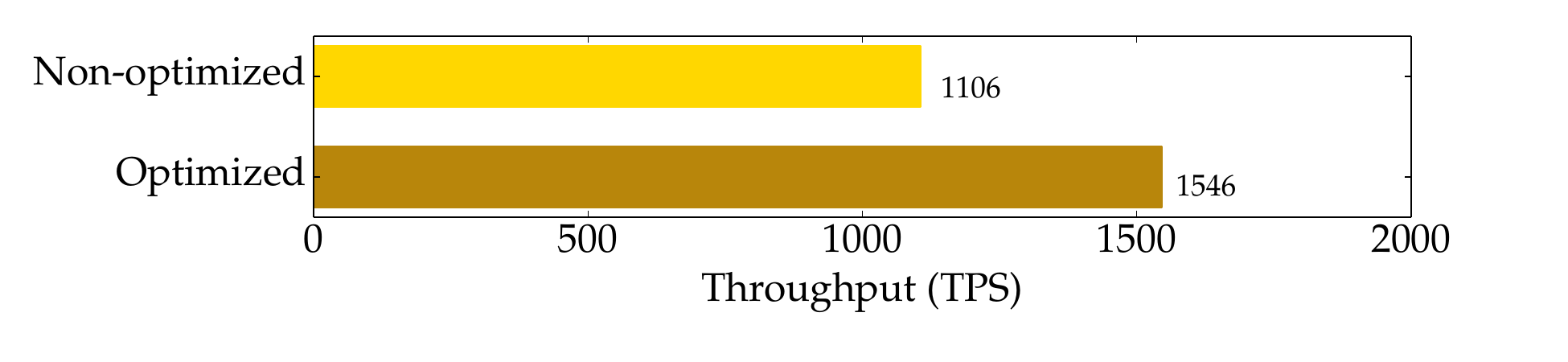}
	\caption{Performance difference when processing each block of a superblock at a time (optimized) and when processing the entire superblock at once (non-optimized)}
	\label{fig:superblock}
\end{figure}

Figure~\ref{fig:superblock} compares the performance obtained with \solution (superblock optimized) and with the naive approach (superblock non-optimized). 
The throughput of \solution (superblock optimized) is 1546\,TPS and 39.7\% higher than the throughput of the naive approach.
This is because trying to persist a large superblock that comprises 15000 transactions leads to high I/O congestion (Section~\ref{superblockopt-section}).

Note that a single c5.2xlarge instance executing an \solution VM in Section~\ref{line:evaloptimizations} yielded a throughput of 1025\,TPS, less than 1546\,TPS for four c5.2xlarge \solution nodes with the superblock optimized benchmark (Figure~\ref{fig:superblock}). The improvement of the latter is due to being able to receive more transactions concurrently at multiple \solution nodes allowing a large number of transactions to be committed in a superblock per consensus.



\subsection{Effects of validation reduction}\label{sec:vreduction}

In order to assess the impact of validation reduction we measured both the total time the \solution VM $\Delta^{n}_{\solution}$ spent treating $k$ native payment transactions and the average time 
$\delta^{n}_{\solution}$ spent eagerly validating these $k$ transactions on a single \solution node executing on 4\,vCPUs and 8\,GB RAM AWS instance (c5.xlarge). As such, the time not affected by the validation optimization is $\beta = \Delta^{n}_{\solution} - \delta^{n}_{\solution}$. Based on this measurement, we could deduce the time $\delta_{\solution}$: the \solution node would spend validating eagerly with the validation optimization if the number of \solution nodes were $n$:
$\delta_{\solution_{v-opt}} =\delta^{n}_{\solution}$/$n$. We know that the \solution VM would spend $\Delta_{\solution} = \beta + \delta_{\solution}$ to treat $k$ transactions. By contrasts, the validation optimized \solution VM would spend $\Delta^{n}_{\solution_{v-opt}} = \beta + \delta^{n}_{\solution}$/$n$. This means the \solution VM validation optimization latency improvement is:
$$S = \frac{\Delta_{\solution} - \Delta^{n}_{\solution_{v-opt}}}{\Delta^{n}_{\solution_{v-opt}}} = \frac{\delta_{\solution}-\delta^{n}_{\solution}/n}{\beta+ \delta^{n}_{\solution}/n}.$$

As $n$ tends to infinity, we thus have a slowdown of: $$\lim_{n\rightarrow +\infty}S = \frac{\delta_{\solution}}{\beta}.$$

Our measurement obtained with $k=15000$ transactions and $n=1$ revealed that 
$\delta^{n}_{\solution} = 4.27$ seconds and $\Delta^{n}_{\solution} = 17.87$ seconds. Hence, we have $\beta =  17.87-4.27 = 13.6$. If we consider $n=4$, $\delta_{\solution_{v-opt}} = 4.27/4 = 1.07$, so that $\Delta^{n}_{\solution_{v-opt}} = 13.6 + 1.07 = 14.67$. The improvement with validation optimization is thus, $(17.87 - 14.67)\cdot100/14.67=21.8\%$. As $n$ tends to infinity, the improvement in performance will be $4.27\cdot100/13.6=31.39\%$

\begin{table}[ht!]
\centering
\setlength{\tabcolsep}{7pt}
\begin{tabular}{ccc} 
 \toprule
 Machine  & Improvement on & Asymptotic gain \\  
 specification & small networks & on large networks \\
 \midrule
 c5.xlarge & 21\% & 31\% \\  
 c5.2xlarge & 24\% & 36\% \\ 
 c5.4xlarge & 28\% & 42\%  \\
 c5.9xlarge & 32\% & 48\%  \\
 \bottomrule
\end{tabular}
\caption{The validation reduction allows the \solution VM to scale horizontally (with the amount of resources allocated to each node) and vertically (with the number 
of nodes where it runs)}
\label{table:vreduction}
\end{table}

Table~\ref{table:vreduction} presents the performance improvements induced by the validation reduction when sending 15000 native payment transactions to the \solution VM deployed on different networks.
Its second column indicates the gain obtained in the smallest instance of our network, when $n=4$, whereas its third column indicates the gain obtained in the largest possible network, when $n\rightarrow \infty$.
The validation reduction allows the performance of the \solution VM to scale vertically because its gain grows from 31\% to 48\% (last column) with the amount of allocated resources and horizontally because its gain grows from 32\% to 48\% (last row) with the size of the network.

\begin{figure}[ht]
	\hspace{1em}
	\includegraphics[scale=0.8]{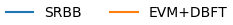}\\
	\includegraphics[scale=0.35,clip=true,viewport=-8 50 700 270]{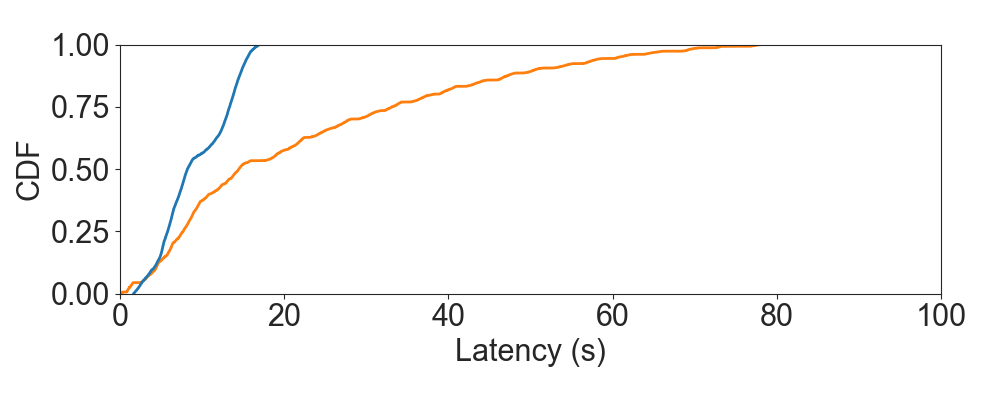}\\
	\includegraphics[scale=0.363,clip=true,viewport=17 25 700 270]{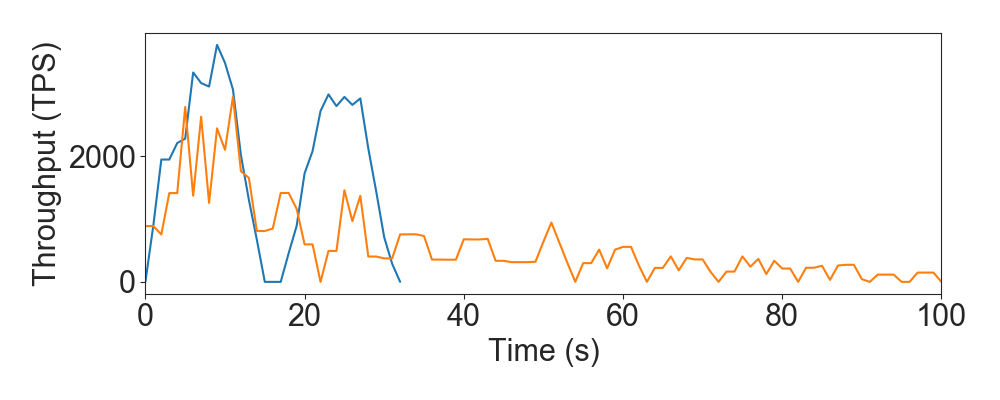}
	\caption{The latency (top) and throughput (bottom) comparison between \solution (optimized) and the default \texttt{geth} EVM equipped with DBFT to decide superblocks (non optimized)
	\label{fig:optimized}}
\end{figure}

\begin{figure*}[ht]
	\hspace{1em}\includegraphics[scale=0.25]{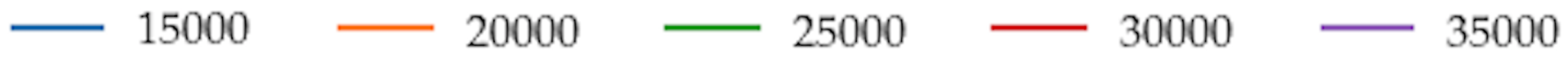}\\
	\includegraphics[scale=0.35]{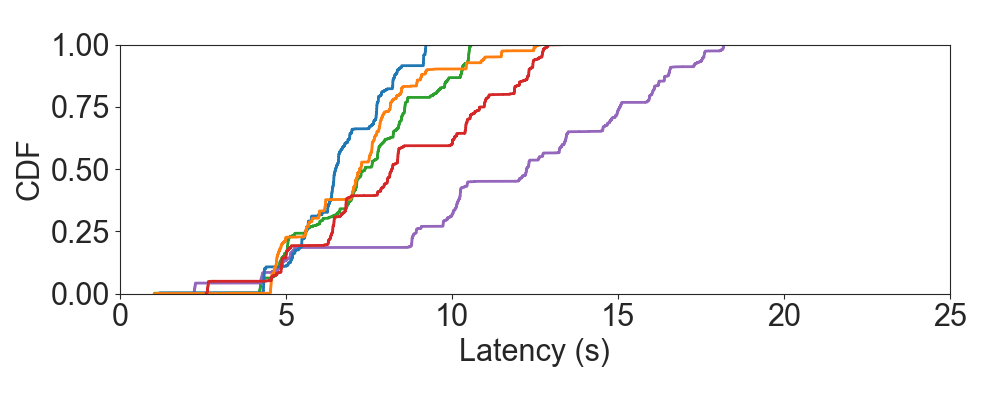}\hspace{-1em}
	\includegraphics[scale=0.35]{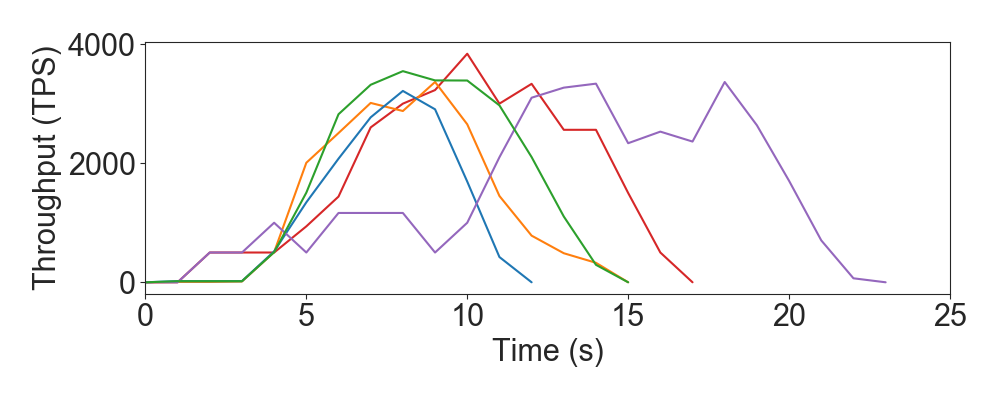}
	\caption{Scalability of \solution as its latency (left) and throughput (right) under varying sending rates (from 15,000\,TPS to 35,000\,TPS) when deployed on 200 machines spread in 5 continents\label{fig:srate}}
\end{figure*}

\subsection{Overall performance improvement}
\label{sec:cumulativeperf}
To show the performance benefit of \solution, we evaluated the naive EVM simply plugged to the superblock supported DBFT consensus (i.e., non-optimized) against \solution (optimized) -- the \solution VM plugged to the superblock supported DBFT consensus. We use a sending rate of 30,000\,TPS with 60,000 transactions sent over 2 seconds to perform this evaluation for four c5.2xlarge nodes.

Figure~\ref{fig:optimized} (top) presents the latency distribution as a cumulative distribution function for all 
transactions for \solution and the non-optimized alternative. 
On the one hand, we observe that the naive EVM-DBFT blockchain takes 23 seconds on average to commit a transaction
and more than 1mn13s to commit some transactions.
The default EVM acts as a bottleneck: while DBFT is optimized to treat transactions fast~\cite{CNG21}, the EVM is not. 
The reason is that the EVM was originally designed for Ethereum, which anyway recommends users to wait for about 
3 minutes (12 confirmations, each taking 15 seconds in expectation) to consider their transaction committed.
On the other hand, \solution commits transactions within 9.6 seconds on average while all transactions are 
committed in less than 16 seconds after they are issued. 
As a result, \solution halves the average latency of the naive blockchain and reduces its observed worst-case latency by $4.6\times$.

Figure~\ref{fig:optimized} (bottom) depicts the throughput over time of the naive blockchain and \solution.
Note that even though \solution transaction latencies are lower than 20 seconds, the throughput of \solution does not drop to zero after 20 seconds: this is because transactions are not all issued at the beginning of the experiments, some of them are issued later.
\solution peaks at a throughput of \textasciitilde4000\,TPS higher than the \textasciitilde3000\,TPS peak throughput that the naive approach reaches. 
The naive EVM+DBFT blockchain and \solution deliver an average throughput of 612\,TPS and 1935\,TPS, respectively. This represents an improvement of $3.2\times$.
Note that the overall speedup of \solution does not correspond to the sum of the speedups obtained by each of our individual optimizations. This is because optimizing the transaction management of \solution
leads to delivering superblock faster to the execution engine. As a result, its transaction management may bottleneck by concurrently receiving, validating, and proposing transactions to the consensus protocol while executing fast-delivered superblocks. 


\subsection{World-wide scalability}\label{sec:geodistributed}

We define scalability as the ability to maintain reasonably good performance when a blockchain has a large number of nodes in a geo-distributed environment.
To evaluate the scalability of \solution, we deployed 200 c5.2xlarge nodes of \solution in 10 regions spanning 5 continents. The benchmark was performed with varying sending rates from 15000 to 35,000\,TPS. We achieved an average throughput of 2000\,TPS over all these executions. The peak throughput was always above 3000\,TPS but showed the best performance at a sending rate of 30,000\,TPS, peaking at close to \textasciitilde4000\,TPS as can be seen on Figure~\ref{fig:srate}. Generally, with the sending rate increasing the peak throughput showed an increase, although it showed a slight dip at a sending rate of 35000\,TPS due to \solution not processing transactions at the send rate due to a bottlenecked CPU and memory.

The latency on the other hand increased with increasing sending rates, with the best latency at a sending rate of 15000\,TPS with all transactions being committed under 10 seconds.



\subsection{Comparison with other blockchains}
\label{section:comparison}

To evaluate \solution in a realistic setting,
we evaluate the performance of a realistic DApp running on \solution and on 6 other state-of-the-art blockchains that also support DApps: Algorand, Avalanche (C-Chain), Ethereum (with Clique consensus), Libra/Diem, Quorum (with IBFT consensus) and Solana. To this end, we used the Diablo benchmark~\cite{BGG21} that is publicly available~\cite{DiabloLink}
and aims at comparing blockchains fairly. We selected the Diablo benchmark because it is, as far as we know, the only 
blockchain benchmark that features a real workload trace. More precisely, we ran the exchange DApp available in Solidity, PyTeal, and Move in the Diablo benchmark that executes the real \textsc{Nasdaq} stock trade trace for 2 minutes (Move is the smart contract language supported by Libra-Diem, PyTeal is the smart contract language used to produce Teal bytecode for Algorand and the Solidity language is supported by all the other blockchains we tested.). Note that we did not include the Red belly blockchain in this benchmark as it does not support smart contracts/DApps.
We used the same deployment setting as in Section~\ref{sec:geodistributed} with 200 c5.2xlarge machines distributed in 10 countries on 5 continents.

\begin{figure*}[ht]
	\includegraphics[scale=0.79]{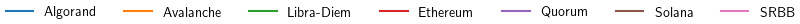}\\
    \includegraphics[scale=0.35]{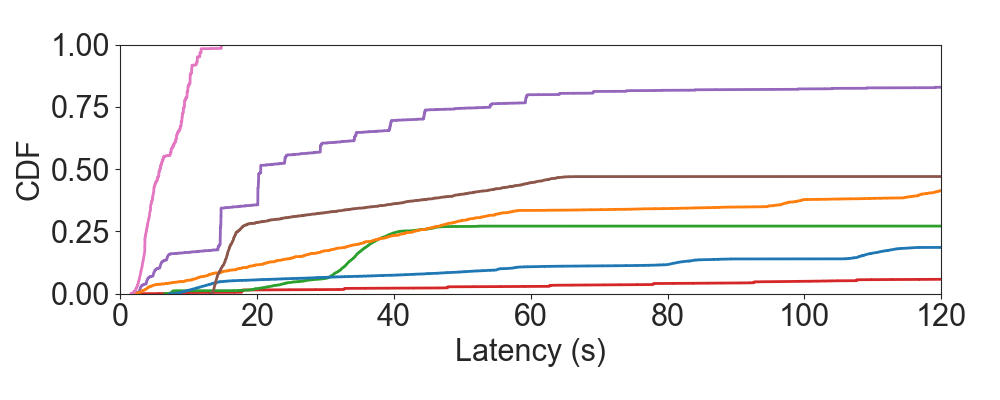}
    \hspace{-1em}\includegraphics[scale=0.35]{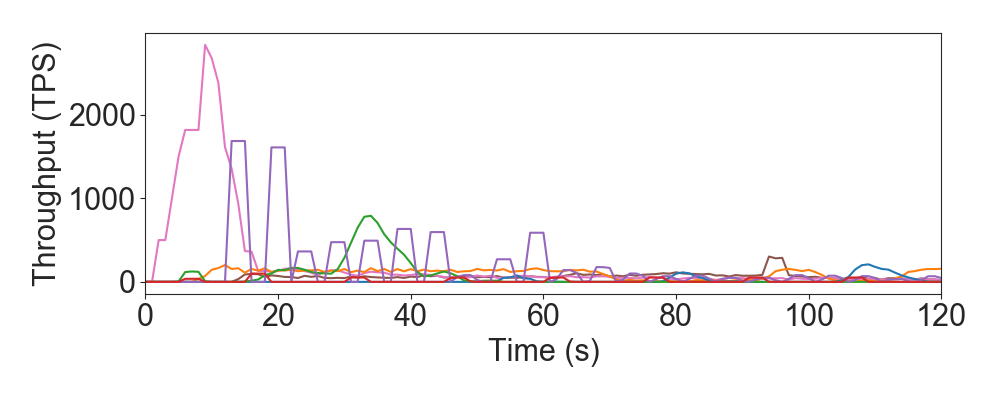}
	\caption{Comparison of modern blockchains latency (left) and throughput (right) when running the exchange decentralized application (DApp) with the \textsc{Nasdaq} trace\label{fig:gafam}}
\end{figure*}

Figure~\ref{fig:gafam} depicts the latency distribution as the cumulative distribution function of the latencies of transactions observed during the experiments and the throughput as the volume of transactions committed per second. 

As expected, \solution is much faster at committing transactions than other blockchains. This happens for several reasons. First, the validation reduction optimization reduces the transaction validation time, allowing transactions to be committed faster. Second, the sub-block storing optimization significantly reduces I/O reducing the waiting time of the CPU before processing transactions.

What is more surprising is that \solution is the only blockchain to not lose any request.
This can be observed in Figure~\ref{fig:gafam} (left) as the ratio of committed transactions over all requested transactions only reaches 1 for \solution. Note that running the benchmark for more than 2 minutes would not allow other blockchains to commit all transactions as we can clearly detect some asymptotic bounds.
The reason is that the \textsc{Nasdaq} workload from Diablo features the stock trades corresponding to Apple, Amazon, Facebook, Google, and Microsoft whose demand is particularly high when the market opens. As other blockchains take longer to handle each request, their backlog of pending transactions increases, and they become saturated. This is well illustrated by the throughput graph on Figure~\ref{fig:gafam} (right), where \solution throughput peaks well above and earlier than other
blockchains.

\section{Related Work}
\label{sec:rw}
In this section, we focus on work enhancing transaction management to improve the performance of DApps.
These works include executing transactions concurrently, introducing faster runtimes, using different language implementations of VMs, and decoupling transaction management tasks. None to our knowledge present the same optimization as us.

\paragraph{Concurrent Execution of Transactions}

Saraph and Herlihy~\cite{EmpStudy} perform an empirical study on speculative concurrent executions
of Ethereum smart contracts by sampling Ethereum blocks and executing smart contract transactions in them concurrently. 
Failed transactions abort and execute sequentially afterward.
This approach is known as an execute-order approach. The authors estimate the speedup of execution by the gas cost of transactions
without calculating the exact performance improvement since they do not implement an EVM using their speculative concurrent approach.
The presented speculative concurrent execution approach is estimated to speedup
performance starting from 8-fold to 2-fold when the number of conflicting transactions gradually increases.
Considering that Ethereum's average throughput is 15\,TPS the speculative concurrent approach would only yield an average throughput of 120\,TPS whereas our average throughput is \textasciitilde{}x17 higher.

Block-STM~\cite{https://doi.org/10.48550/arxiv.2203.06871} presents an in-memory execute-order paradigm smart contract parallel execution engine built on the principles of Software Transaction Memory (STM). It dynamically detects and avoids conflicts during concurrent transaction execution reaching the same final state as a sequential execution. Block-STM is proven safe and live and is implemented on the Diem blockchain and delivers
a throughput of 140k\,TPS in a baseline benchmark and 60k\,TPS in a contended workload. Block-STM executes conflicting transactions by keeping multiple versions of transaction executions in an in-memory data structure where each version corresponds to a unique location in the execution order of all transactions. The correct version is selected by a validation method and finalized. Since we do not implement concurrent execution of transactions, our approach does not require maintaining duplicated transaction execution versions and then a validation method to resolve those versions, all of which increases the workload on the CPU.

Similar to Block-STM~\cite{https://doi.org/10.48550/arxiv.2203.06871}, Parwat et al.~\cite{anjana2020efficient} proposes a concurrent smart contract execution approach using STM. However, in contrast, the latter follows an order-execute paradigm where miners create a Block Graph (BG) for conflicting transactions in a block and validators execute the smart contract transactions concurrently based on the BG. By having the distinction between miners and validators Parwat et al.~\cite{anjana2020efficient} could have issues related to byzantine miners~\cite{https://doi.org/10.48550/arxiv.2203.06871}. By pre-computing a BG, there is additional overhead compared to Block-STM. Tracking dependencies and building a BG before execution can be a tedious task and could seriously impair performance when the percentage of conflicting transactions is high. Our work on the other hand can maintain performance despite a high percentage of conflicting transactions since we execute transactions sequentially without any dependency pre-computations.

ParBlockchain~\cite{8885058} performs concurrent executions of distributed applications following an order-execute paradigm coined OXII.  Similar to Parwat et al.~\cite{anjana2020efficient}, ParBlockchain orders transactions within a block based on a dependency graph prior to execution. OXII is integrated into a permissioned blockchain coined ParBlockchain that features three types of nodes. Namely, 1) Clients that send requests to the blockchain;
2) Orderers that agree on the order of transactions, builds a dependency graph and executes consensus (e.g. PBFT); 3) Executors
that validate and execute transactions. The transactions that are found to have no dependencies in the dependency graph execute concurrently.
The evaluation reports better performance than sequential order-execute (OX) and execute-order-validate (XOV) for workloads with any degree of contention. ParBlockchain's decoupling tasks to three types of nodes reduces fault tolerance. Orderer nodes can tolerate up to $f$ failed nodes s.t. $N/3 > f$ but since these nodes are connected to decoupled executor nodes, the fault tolerance of executor nodes should be considered separately. This could reduce the fault tolerance of the entire blockchain. In contrast, \solution does not have this drawback brought forth by decoupling, and assumes a fault tolerance of $f$ at most for the entire network.

Kim et al.~\cite{273865} propose an off-chain concurrent execution of Ethereum transactions for scalable testing and profiling of smart contracts. The proposed approach lets a {\ttt geth} node synchronize and record the Ethereum state prior to every transaction execution known as the input sub-state and after the execution of every transaction known as the output sub-state. Thus, an input and output sub-state are recorded for each transaction in the Ethereum trace.
A sub-state of a transaction contains the information required to execute the transaction in complete isolation from other transactions.
After the sub-state recording is complete, 
all transactions in the Ethereum trace are executed in isolation of one another, each transaction executing in a dedicated thread concurrently. 
Each transaction is applied to its recorded input sub-state and the resulting output sub-state is compared with the recorded output sub-state to validate the execution. Although ideal for testing purposes after recording sub-states, unlike \solution, this approach cannot work in a blockchain receiving transactions in real-time. This is because recording transaction sub-states is impossible without actually executing transactions from clients sequentially.

\paragraph{Other VM optimizations}
Ethereum introduced eWASM (Ethereum Web Assembly) as a runtime for transaction executions~\cite{eWASMdoc}. Although eWASM bytecode is deemed as a faster alternative to EVM bytecode for smart contract (i.e., DApp) executions~\cite{VmMatters}, recent work comparing the performance of WASM EVMs and EVMs~\cite{ewasm} show that eWASM integrated VMs do not outperform {\ttt geth} and {\ttt Openethereum}. {\ttt Geth} shows
the best performance~\cite{VmMatters} for smart contract executions among popular Ethereum clients.
The authors claim that declined performance of the eWASM VMs is due to additional overhead for gas metering and inefficient context switching of methods used for the interaction of high-level opcodes of Ethereum with WASM.

Different implementations of Ethereum yield varying performances. Rouhani and Deters~\cite{8342866} compare the performance of {\ttt geth} and the Rust implementation of Ethereum called Parity (now known as {\ttt Openethereum}) in two private permissioned blockchains, one consisting only {\ttt geth} nodes and the other only {\ttt Openethereum} nodes.
As reported, Parity processes transactions 89.8\% faster than {\ttt geth}. This performance disparity is questionable, since, in contrast, paper~\cite{ewasm} shows Geth's contract execution performance to be faster than Openethereum's.
\paragraph{Popular blockchains with enhanced transaction management}

RainBlock~\cite{ponnapalli2021rainblock} enhances the transaction management of Ethereum by introducing optimizations that remove the I/O bottleneck in the transaction execution phase caused by state trie updates. It removes I/O
from the critical path by decoupling transaction executions from  I/O by sharding its concurrent read-write Merkel Trie (DSM tree) into sub-trees and storing them in the memory of separate storage nodes. I/O helpers pre-fetch Merkel trie data from storage nodes and submit these trie data to miners for transaction execution.
Miners update the trie data and update the storage nodes asynchronously. While the reported performance of RainBlock is 20K\,TPS for 32\,GB RAM, 8\,vCPU instances spanning 4 AWS regions which include 4 miners, 16 I/O helpers, and 16 storage nodes, the decoupling made to store tries in separate storage nodes adversely impacts its fault tolerance. If storage nodes are byzantine or compromised with a DDoS attack, the safety of the blockchain could be compromised as miners will not receive the Merkel trie data required for execution. In contrast, \solution does not have this problem as there is no decoupling of \solution nodes into separate storage nodes. 


\sloppy{Solana~\cite{yakovenko2018solana} is a high-performance blockchain providing enhanced transaction management. Perhaps the most noteworthy
of its transaction management optimizations is Sealevel~\cite{SolanaSealevel} -- a parallel smart contract runtime made possible by the unique transactions of Solana that describe read and write states separately. This allows non-conflicting transactions and transactions that read from the same state to execute concurrently. Solana also optimizes transaction validation by pipelining validation in a software-built transaction processing unit (TPU) into four stages. The signature verification is offloaded to a GPU. However, by XORing its state hash, Solana has been found to be vulnerable to state hash collision attacks~\cite{SolanaWeak} whereas \solution does not suffer from this vulnerability.}
\vspace{-0.5em}

\section{Conclusion}
\label{sec:conclusion}
In this paper, we enhanced transaction management by reducing transaction validations and optimizing block processing.
Through these enhancements, we developed the \solution VM. We then rigorously evaluated the performance increase of each of our optimizations.
To apply our contributions, we developed \solutionlong supporting the largest eco-system of DApps. Finally, we empirically showed that \solution outperforms recent high-performance blockchains that support DApps.

\appendix

\section{Consensus Protocol and Proofs of Correctness}\label{line:proof}
The protocol is divided in two procedures, {\ttt start\_new\_consensus} at lines~\ref{line:start-new-cons}--\ref{line:superblock} that spawns a new instance of (multivalue) consensus by incrementing the replicated state machine {\ttt index}, and {\ttt consensus\_propose} at lines~\ref{line:cons-start}--\ref{line:cons-end} that ensures that the consensus participants find an agreement on a superblock comprising all the proposed blocks that are acceptable. The idea of {\ttt consensus\_propose} builds upon classic reduction~\cite{BCG93,BKR94} by executing an all-to-all reliable broadcast~\cite{B87} to exchange $n$ proposals, guaranteeing that any block delivered to a correct process is delivered to all the correct processes: any delivered proposal is stored in an array $\ms{proposals}$ at the index corresponding to the identifier of the broadcaster.

\begin{lstlisting}[language=parameterized,basicstyle=\LSTfont,escapechar = ?,escapeinside={(*}{*)},frame = single,firstnumber=20]
index := 0 // consensus instance
blockQueue := (*$\emptyset$*) // pending block to propose
commitChan := chan []byte  // commit channel

start_new_consensus(): (*\label{line:start-new-cons}*)
  index := index + 1 // increment round
  myBlock := blockQueue.peek() // get block proposal
  superblock := consensus_propose(myBlock) // block (*\label{line:propose}*)
  if (myBlock is in superblock) then    (*\label{line:remove-from-queue-start}*)
    blockQueue.poll() // dequeue proposal   (*\label{line:remove-from-queue-end}*)
  commitChan := superblock // send to gRPC srvce (*\label{line:superblock}*)
\end{lstlisting}

A binary consensus at index $k$ is started
with input value {\ttt true} for each index $k$ where a block proposal has been recorded (line~\ref{line:bbc}). 
To limit errors, \solution uses the formally verified deterministic binary consensus of DBFT~\cite{CGLR18}, we omit the pseudocode for the sake of space and refer the reader to the formal verification of the protocol~\cite{TG19,BGK21}. 

As soon as some of these binary consensus instances return 1, the protocol spawns binary consensus instances with proposed value $\lit{false}$ for each of the non reliably delivered blocks at line~\ref{line:bbc-propose}. Note that this invocation is non-blocking.
As the reliable broadcast fills the {\ttt block} in parallel, it is likely that the blocks reliably broadcast by correct processes have been reliably delivered resulting in as many invocations of the binary consensus with value {\ttt true} instead. Once all the $n$ binary consensus instances have terminated, i.e., {\ttt decidedCount == n} at line~\ref{line:bbc-end}, the superblock is generated with all the reliably delivered blocks for which the corresponding binary consensus returned {\ttt true} (lines~\ref{line:supblock-start}--\ref{line:cons-end}).
At the end of {\ttt start\_new\_consensus}, if the superblock of the consensus contains the block proposed, then this block is removed from the {\ttt blockQueue} at lines~\ref{line:remove-from-queue-start} and~\ref{line:remove-from-queue-end} to avoid reproposing it later.

We show that \solution solves the blockchain problem (Def.~\ref{def:blockchain}).

\begin{lemma}
	\label{theorem1}
	If at least one correct \solution node $\lit{consensus-propose}$s to a consensus instance $i$, then 
	every correct \solution node decides on the same superblock at consensus instance $i$.
\end{lemma}

\begin{proof}
If a correct \solution node $p$ $\lit{consensus-propose}$s, say $v$, to a consensus instance $i$, then $p$ reliably broadcast $v$ at line~\ref{line:rbcast}. By the reliable broadcast properties~\cite{B87}, we know that $v$ is delivered at line~\ref{line:rb-deliver} at all correct validator nodes. By assumption, there are at least $2f+1$ correct proposers invoking the reliable broadcast, hence all correct proposers eventually populate their $\lit{block}$ array with at least one common value. All correct proposers will thus have input $\lit{true}$ for the corresponding binary consensus instance at line~\ref{line:bbc}.
Now it could be the case that other values are reliably-broadcast by byzantine nodes, however, reliable broadcast guarantees that if a correct proposer delivers a valid value $v$, then all correct proposers deliver $v$. By the validity and termination 
properties of the DBFT binary consensus~\cite{CGLR18}, the decided value for the binary consensus instance at line~\ref{line:bbc-true} is the same at all correct \solution nodes. It follows that all correct \solution nodes have the same bit array of $\lit{decBlocks}$ values at line~\ref{line:decided-blocks} and that they all return the same superblock at line~\ref{line:cons-end} for consensus instance $i$.
\end{proof}

\begin{lstlisting}[language=parameterized,basicstyle=\LSTfont,escapechar=|,escapeinside={(*}{*)},frame = single,firstnumber=31]
blocks := (*$\emptyset$*) // blocks delivered by reliable bcast
upon reliable_broadcast.deliver(i, block): (*\label{line:rb-deliver}*)
    blocks[i] := block // append block to list
    decBlocks[i] := b_consensus.propose(i, true) (*\label{line:bbc}*)
consensus_propose(myBlock):  (*\label{line:cons-start}*)
  decCount := 0 // # decided bin. cons. instances
  decBlocks :=  (*$\emptyset$*) 
  reliable_broadcast.broadcast(myId, myBlock)  (*\label{line:rbcast}*)
  wait until (*$\exists$*)i : b_consensus.decide(i) == true (*\label{line:bbc-true}*)
     for j  from 0 to n do
       if blocks[j] is null then
         decBlocks[j] := b_consensus.propose(j, false) (*\label{line:bbc-propose}*)
     decCount := decCount+1
   wait until decCount == n (*\label{line:bbc-end}*) 
     superblock := (*$\emptyset$*) (*\label{line:supblock-start}*)
     for i from 0 to n do 
       if decBlocks[i] is true then (*\label{line:decided-blocks}*)
         superblock.add(blocks[i])
  return superblock  (*\label{line:cons-end}*)
 \end{lstlisting}


The next three theorems show that \solution satisfies each of the three properties of the blockchain problem (Definition~\ref{def:blockchain}).
\begin{theorem}
	\solution satisfies the safety property.
\end{theorem}
\begin{proof}
	The proof follows from the fact that any block $B_{\ell}$ at index $\ell$ of the chain is identical for all correct \solution nodes due to Lemma~\ref{theorem1}.
	Due to network asynchrony, it could be that a correct node $p_1$ is aware of block $B_{\ell+1}$  at index $\ell+1$, whereas another correct node $p_2$ has not created this block $B_{\ell+1}$ yet.
	At this time, $p_2$ maintains a chain of blocks that is a prefix of the chain maintained by $p_1$.
	And more generally, the two chains of blocks maintained locally by two correct blockchain nodes are either identical or one is a prefix of the other. 
\end{proof}

\begin{theorem}
	\solution satisfies the validity property.
\end{theorem}	
\begin{proof}
By examination of the code at line~\ref{line:validate}, only valid transactions are executed and persisted to disk at every correct \solution node.
It follows that for all indices $\ell$, the block $B_{\ell}$ is valid.
\end{proof}

\begin{theorem}\label{thm:liveness}
	\solution satisfies the liveness property.
\end{theorem}	
\begin{proof}
	As long as a correct \solution nodes receives a transaction, we know that the transaction is eventually proposed by line~\ref{line:propose}. 
	The proof follows from the termination of the consensus algorithm~\cite{BGK21} and the fact that \solution keeps spawning new consensus instances as long as correct \solution nodes have pending transactions.
\end{proof}

\bibliographystyle{ACM-Reference-Format} 
\bibliography{reference} 

\end{document}